\documentclass[a4paper,12pt]{article}

\usepackage{a4wide}

\usepackage{setspace}
\usepackage{algcompatible}
\usepackage{amsmath,amsthm,amssymb,amsfonts,authblk,algorithm}
\usepackage{xspace,graphicx}
\usepackage[hypertex]{hyperref}
\title{Lower bounds on the M\"{u}nchhausen problem}
\author{Michael Brand\\
\texttt{michael.brand@alumni.weizmann.ac.il}}
\affil{Faculty of IT, Monash University\\
Clayton, VIC 3800\\
Australia}
\date{\today}

\def\mathbi#1{\textbf{\em #1}}
\newcommand{\vn}{\vec{\mathbi{n}}}

\newcommand{\M}{M\"{u}nchhausen\xspace}
\renewcommand{\O}{\mathop{\mathrm{O}}}
\newcommand{\Om}{\mathop{\mathrm{\Omega}}}

\newcommand{\ceil}[1]{\left\lceil #1 \right\rceil}

\newcommand{\defeq}{\stackrel{\text{def}}{=}}

\newtheorem{thm}{Theorem}

\begin{document}
\maketitle

\begin{abstract}
``The Baron's omni-sequence'', $B(n)$, first defined by Khovanova and Lewis
(2011), is a sequence that gives for each $n$
the minimum number of weighings on balance
scales that can verify the correct labeling of $n$ identically-looking
coins with distinct integer weights between $1$ gram and $n$ grams.

A trivial lower bound on $B(n)$ is $\log_3 n$, and it has been shown that
$B(n)$ is $\log_3 n + \O(\log \log n)$. In this paper we give a first
nontrivial lower bound to the \M problem, showing that there is an
infinite number of $n$ values for which $B(n)\neq \lceil\log_3 n\rceil$.

Furthermore, we show that if $N(k)$ is the number of $n$ values for which
$k=\lceil\log_3 n\rceil$ and $B(n)\neq k$, then $N(k)$ is an unbounded
function of $k$.
\end{abstract}

\section{Introduction}
Coin-weighing puzzles have been abundantly discussed in the mathematical
literature over the past 60 years
(see, e.g.\ \cite{Steinhaus:snapshots,Halbeisen:general_counterfeit,
Smith:counterfeit,Dyson:Pennies}). In coin-weighing problems one must typically
identify a counterfeit coin from a set of identically-looking coins by use of
balance scales, utilizing the knowledge that the counterfeit coin has
distinctive weight. This can be generalized to the problem of identifying a
coin, or a subset of the coins, based on distinctive weight characteristics,
or, alternatively, to the problem of establishing the weight of a
given coin.

This paper relates to a different kind of coin-weighing puzzle, which we call
``The \M coin-weighing problem'' (following, e.g., \cite{Brand:Munchhausen2}).
Consider the following question: given $n$ coins with distinct
integer weights between $1$ gram and $n$ grams, each labeled by a distinct
integer label between $1$ and $n$, what is the minimum number of weighings
of these $n$ coins on balance scales that can prove unequivocally that all
coins are labeled by their correct weight?

This question differs from classic coin-weighing problems in that we do not
need to \emph{discover} the weights, but only to determine whether or not
a given labeling of weights is the correct one. To establish the weights one
would require $\Om(n \log{n})$ weighings (as can be proved by reasoning similar
to that which establishes lower bounds for comparative sorting
\cite{Knuth:programming,Cormen:algorithms}), whereas merely
verifying an existing labeling can be performed trivially in $\O(n)$
weighings.

This question, inspired by a riddle that appeared in the Moscow Mathematical
Olympiad \cite{Tokarev:Olympiad}, gives rise to an integer sequence, $B(n)$,
that was studied in \cite{Khovanova:Baron} and was dubbed there ``The Baron's
omni-sequence''. It appears as sequence A186313 in the On-line Encyclopedia of
Integer Sequences \cite{A186313:Omnisequence}.

Though much progress has been made to tighten the known upper bounds on
$B(n)$ \cite{Khovanova:Baron, Brand:Tightening, Brand:Munchhausen2}, the trivial
lower bound of $\log_3 n$ has proved surprisingly resilient. This lower bound
stems from the straightforward observation that if the number of weighings is
less than $\log_3 n$, there must be at least two coins that participate in all
weighings in identical roles. (For each weighing, they are either both on the
left-hand side of the scales, both on the right-hand side or both held out from
the weighing.)
This being the case, the weights of the two coins can be exchanged with no
change to the outcome of any of the weighings, and therefore the weighings
cannot provide an unequivocal verification of the weights.

In this paper we present a first nontrivial lower bound for this problem.
Namely, we prove the following theorem.

\begin{thm}\label{T:lower}
For any $n$,
\begin{equation}\label{Eq:loglog}
3^{B(n)} \ge n+\Om(\log \log n).
\end{equation}
Equivalently,
\begin{equation}\label{Eq:log}
N(k)\in\Om(\log k),
\end{equation}
where $N(k)$ is the number of $n$ values for which
$k=\lceil\log_3 n\rceil$ and $B(n)\neq k$.
\end{thm}

\section{Proof of the main theorem}\label{S:proof}

We begin by introducing some terminology. First, following
\cite{Brand:Munchhausen2}, we describe sequences of weighings by means of
matrices.
A $k \times n$ matrix, $M$, whose elements, $M_{ij}$ belong to the set
$\{-1,0,1\}$, describes a sequence of $k$ weighings of $n$ coins. If
$M_{ij}=1$, this indicates that coin $j$ is to be placed on the right hand
side of the scales on the $i$'th weighing. If it is $-1$, the coin is to be
placed on the left hand side. A ``$0$'' indicates that on the $i$'th weighing
the coin is to be held out.

In the case of the \M problem, it is known what weights the coins to be
weighed are: the first coin weighs $1$ gram, the second weighs $2$ grams,
and so on. We describe these weights by the vector $\vn=[1,\ldots, n]^T$.
The result of the weighing sequence is therefore given by the element-wise signs
of the vector $M\vn$. We describe the operation including both multiplication
by $\vn$ and sign-taking by the single operator $w(M)$.

A matrix is \emph{\M} if the sequence of weighings it describes generates a
sequence of weigh results (signs) when weighing $\vn$ that is unique among
all possible permutations of $\vn$. Equivalently, a matrix is \M if
$w(M)=w(M\pi) \Rightarrow \pi=I$ for an $n \times n$ permutation matrix $\pi$.

The Baron's omni-sequence is the sequence that gives for each $n$ the
minimum $k$ for which there exists a $k \times n$ \M matrix.

Theorem~\ref{T:lower} gives a first nontrivial lower bound on $k$. We prove
it now.

\begin{proof}[Proof of Theorem~\ref{T:lower}]
Consider, first, the trivial lower bound for the Baron's omni-sequence. In
matrix terminology, we claim that if a $k\times n$ matrix, $M$, is \M, then
$n \le 3^k$. The reason for this is that if $n>3^k$, at least two of $M$'s
columns are identical. A permutation $\pi$ permuting the columns of $M$
by switching identical columns will have no effect on it: we have
$M=M\pi$, and therefore necessarily also $w(M)=w(M\pi)$.

The relevant observation regarding this proof is that it demonstrates that the
columns of $M$ must be distinct. Because they all belong to the set
$\{-1,0,1\}^k$, of size $3^k$, the set, $C$, of choices for the set of $M$'s
columns (ignoring their order) is limited by $|C|\le\binom{3^k}{n}$.

Consider, now, row permutations on $M$. For an $M$ with a large $k$, there are
many row permutations of $M$ that do not change $w(M)$. For example, consider
that each row of $M$ generates a sign that has only $3$ possibilities. As
such, there will be at least $\lceil k/3 \rceil$ rows that share the same
generated sign. Any $\sigma_1$, $\sigma_2$ of the $(\lceil k/3 \rceil)!$
possible row permutations on $M$ that keep all rows other than these
$\lceil k/3 \rceil$ as fixed points share the same
$w(\sigma_1 M)=w(\sigma_2 M)=w(M)$. We define $R$ to be the set of all row
permutations that satisfy $w(\sigma M)=w(M)$, noting that
$|R|\ge(\ceil k/3 \rceil)!$.

We claim that for $M$ to be \M,
\begin{equation}\label{Eq:CR}
|C|\ge |R|.
\end{equation}

If we define $l=3^k-n$, then $|C|\le\binom{3^k}{n}$ implies $|C|<3^{kl}$. On the
other hand, $\log_3 |R|$ is $\Om(k \log k)$, so Equation~\eqref{Eq:CR} implies
that $l$ is $\Om(\log k)$.

Because a lower bound for $l$ is also a lower bound for $N(k)$ of
Equation \eqref{Eq:log}, Equation~\eqref{Eq:CR} directly implies
Equation \eqref{Eq:log}, which, in turn, implies Equation \eqref{Eq:loglog},
because $k$ is $\Om(\log n)$. In other words, proving Equation~\eqref{Eq:CR}
is tantamount
to proving the entire theorem. We now proceed to establish this claim.

We define the relation $f:R\rightarrow C$ as follows. For $\sigma\in R$,
$f(\sigma)$ is the set of columns of $\sigma M$. Because changing the order
of the weighings clearly has no effect on whether or not a set of weighings
establishes unequivocally the weights of $n$ coins, $\sigma M$ is \M if and
only if $M$ is \M, so by definition the set of columns of $\sigma M$ is
necessarily a member of $C$. Instead of showing Equation~\eqref{Eq:CR}, we
make the stronger claim that $f$ is one-to-one.

To prove this, let us assume to the contrary that $f$ is not one-to-one. This
indicates the existence of two row permutations $\sigma_1,\sigma_2\in R$ for
which $f(\sigma_1)=f(\sigma_2)$. Because the application of a permutation is
invertible, we know that $\sigma_1 M \neq \sigma_2 M$. The two are therefore
related by a column permutation, $\pi$, which is not the identity, as follows:
\[
\sigma_1 M \pi = \sigma_2 M.
\]
Let $\sigma\defeq \sigma_2^{-1} \sigma_1$, then
\[
\sigma M \pi = M.
\]
Recall that by definition of $R$, we have $w(M)=w(\sigma M)$, so
\[
w(\sigma M \pi)=w(M)=w(\sigma M),
\]
so by definition $\sigma M$ cannot be a \M matrix. However, as argued
earlier, $\sigma M$ is \M if and only if $M$ is \M, so the above implies
that $M$, too, is not \M, contradicting the assumption.
\end{proof}

\section{Conclusions}
With the new Theorem~\ref{T:lower}, the best known bounds now place $n$
between $3^k - \Om(\log k)$ and $3^k / \O(\text{polylog}\; k)$, for $n$ to
satisfy $B(n)=k$.
This still leaves a significant window for
further refinement. At the current time, it is not even known whether $B(n)$
is a monotone sequence.


\end{document}